\documentclass[11pt, a4paper]{article}

\usepackage[a4paper,margin=2.5cm]{geometry}
\usepackage[utf8]{inputenc}
\usepackage{booktabs}
\usepackage[small]{caption}
\usepackage{lmodern}

\usepackage{graphicx}
\usepackage{tikz}
\usepackage{algorithm}
\usepackage{algpseudocode}
\usepackage{xspace}
\usepackage{mathtools}
\usepackage[super]{nth}
\usepackage{amsmath}

\usepackage{amsthm}
\usepackage{amsfonts}
\usepackage{amssymb}
\usepackage{thm-restate}
\usepackage{microtype}
\usepackage{hyperref}
\usepackage{enumitem}
\usepackage{multirow}
\DeclareCaptionFont{xbf}{\bfseries\boldmath}
\usepackage{tablefootnote}
\usepackage{todonotes}

\usepackage{svg}

\usepackage[font=scriptsize]{caption}
\usepackage{comment}

\usepackage[sort&compress,numbers]{natbib}
\usepackage{cleveref}
\usepackage{array}
\newcolumntype{x}[1]{>{\centering\arraybackslash\hspace{0pt}}p{#1}}

\algdef{SE}[SUBALG]{Indent}{EndIndent}{}{\algorithmicend\ }%
\algtext*{Indent}
\algtext*{EndIndent}

\algdef{SE}{Upon}{EndUpon}[1]{\textbf{upon} \(\mbox{#1}\) \textbf{do}}{\textbf{end upon}}%

\theoremstyle{plain}
\newtheorem{theorem}{Theorem}[section]
\newtheorem{definition}[theorem]{Definition}
\newtheorem{lemma}[theorem]{Lemma}
\newtheorem{corollary}[theorem]{Corollary}

\newcommand{\qedClaim}{\hfill \ensuremath{\Box}}

\newcommand{\safe}{\ensuremath{\mathrm{Safe}}\xspace}
\newcommand\ballOf[1]{\ensuremath{\mathrm{Ball( \allowbreak\mathnormal{#1})}}\xspace}
\newcommand{\convexHull}{\ensuremath{\mathrm{Conv}}\xspace}
\newcommand{\radius}{\ensuremath{\mathrm{Rad}}\xspace}
\newcommand{\ballMid}{\ensuremath{\mathrm{BallMidpoint}}\xspace}
\newcommand{\ballMidOf}[1]{\ensuremath{\mathrm{BallMidpoint(\allowbreak\mathnormal{#1})}}\xspace}

\newcommand{\Rho}{\mathrm{P}}
\newcommand{\seb}{smallest enclosing ball\xspace}

\newcommand{\mytitle}[1]{
    \begingroup
    \fontsize{16pt}{18pt}\fontseries{bx}\selectfont \centering #1 \par  
    \endgroup
}
\newcommand{\myauthors}[1]{
    \begingroup
    \fontsize{12pt}{14pt}\fontseries{bx}\selectfont \centering #1 \par  
    \endgroup
}

\newcommand{\mykeywords}[1]{
    \begingroup
    \fontsize{10pt}{12pt}\selectfont \textbf{Keywords.} #1   
    \endgroup
}

\begin{document}
\mytitle{Faster Convergence of Multidimensional Approximate Agreement via Smallest Enclosing Balls}

\bigskip
\bigskip

\myauthors{Darya Melnyk\footnote{\texttt{melnyk@tu-berlin.de}, TU Berlin, Germany}}

\bigskip
\bigskip

\begin{abstract}

This work considers the multidimensional approximate agreement problem. In this problem, $n$ parties in a distributed system, up to $t$ of which may be corrupted by a Byzantine adversary, need to output vectors that are close to each other and that lie inside the convex hull of all non-corrupted input vectors. We assume that nodes communicate in a fully-connected authenticated network and analyze synchronous and asynchronous communication models. The focus of this work is on the contraction factor of approximate agreement protocols.

The first multidimensional approximate agreement protocols had a contraction rate of $1-1/n$(VG, PODC'13) and $\sqrt[d]{1/2}$(MH, STOC'13). While a rate below $1$ is sufficient for convergence, it is not sufficient for practical applications. To date, the best known convergence rate of approximate agreement algorithms is $\sqrt{7/8}\approx0.935$ (FN, DISC'18), which is achieved through the MidExtremes protocol. This stands in contrast to the lower bound on the convergence rate in the $1$-dimensional setting, which is $1/2$. 

In this work, we propose \ballMid\ -- a novel approximate agreement protocol with a contraction rate of $1/\sqrt{2}\approx 0.707$ in the synchronous and the asynchronous communication models. This algorithm satisfies the optimal resilience under convex validity. The presented contraction rate is achieved by choosing the midpoint of the smallest enclosing ball of the so-called local safe areas, and it is tight for the presented algorithms. Similar to (FN, DISC'18), our algorithm is coordinate-free, and the point inside the safe area can be computed efficiently. This work presents the first multidimensional approximate agreement protocol where the convergence rate is closer to the best known lower bound rather than the upper bound of $1$.

\end{abstract}

\mykeywords{vector consensus, smallest enclosing ball, convex validity, $1$-center}

\section{Introduction}

Many modern distributed applications handle multidimensional data, such as machine learning~\cite{10.1145/3616537}, collaborative learning~\cite{NEURIPS2021_d2cd33e9}, distributed optimization~\cite{10.1145/3465084.3467902}, robot gathering~\cite{PATTANAYAK2019145}, and large-scale elections~\cite{7474137}. In such systems, agreement protocols can be used to aggregate the local vectors into one global output. Exact agreement protocols are, however, either slow (the number of rounds needed is linear in the number of permitted failures), or even impossible (such as asynchronous deterministic agreement with one crash failure). This motivates using approximate agreement algorithms in practical applications. Approximate agreement allows nodes to converge to a common output. The application can thereby define the accuracy of such algorithms, i.e., how close the final vectors should be. In the particular instance of collaborative learning, the problem is equivalent to approximate agreement. 

To make multidimensional approximate agreement algorithms practical, several properties need to be satisfied: the quality of the agreement vector, the contraction rate, efficient local computability, and resilience. For example, the algorithm that applies one-dimensional approximate agreement in every coordinate is efficient, it has optimal resilience and an optimal contraction rate of $1/2$. However, this algorithm is dependent on a fixed coordinate system, and the output vector is only guaranteed to be inside a coordinate-parallel box around the non-corrupted input vectors. In~\cite{mendes2015multidimensional}, algorithms have been proposed that focus on providing a vector of better quality. In particular, they make sure that the output vectors are inside the convex hull of all non-corrupted input vectors (convex validity). Unfortunately, such algorithms suffer from a resilience that depends on the dimension of the vectors~\cite{mendes2015multidimensional}. Moreover, the best-known contraction rate of such algorithms is $\sqrt{7/8}\approx0.935$, which is still close to $1$, and thus makes the algorithms impractical.

In this work, we take a first step towards making multidimensional approximate agreement algorithms that satisfy convex validity more practical. We propose to use the midpoint of the smallest enclosing ball of the safe area to choose a better vector for the following round. With this strategy, we receive an algorithm with a contraction rate of $1/\sqrt{2}$, and thus get closer to the best known lower bound of $1/2$. This makes the proposed algorithm more suitable to practical applications and mitigates one of its previous drawbacks.

\subsection{Our Contribution}

In the following, we summarize the contributions of our paper:
\begin{itemize}
    \item We provide the first multidimensional Byzantine approximate agreement algorithm under convex validity for the synchronous all-to-all communication model with authenticated channels that achieves a contraction rate of $\frac{1}{\sqrt{2}}$. This algorithm is based on computing the midpoint of the smallest enclosing ball of the safe area, called \ballMid. This algorithm satisfies the optimal resilience bound of $n>\max(3,d+1)t$, where $d$ is the dimension of the input and $t$ is an upper bound on the number of corrupted nodes.
    \item We extend our results to asynchronous communication, showing that a contraction of $\frac{1}{\sqrt{2}}$ can also be achieved in the asynchronous case if $n>\max(3,d+2)t$.
    \item We show that the contraction rate presented in this paper is asymptotically tight for the synchronous and the asynchronous \ballMid algorithms. To this end, we present a worst-case input and Byzantine attack for each communication model. In this construction, the contraction factor remains $\frac{1}{\sqrt{2}}-\delta$ for each iteration of the \ballMid algorithm, where $\delta$ is an arbitrarily small constant.
    \item In the case $d=1$, our algorithms are equivalent to the MidExtremes algorithm, as the diameter of the points and diameter of the ball are the same in $1$D, and thus have an optimal contraction rate of $1/2$.
\end{itemize}

In Table~\ref{tab:comparison-prior-work}, we compare the \ballMid to other choices of a point inside the safe area.

\begin{table}[tbh]
\centering
\begin{tabular}{ l | c c x{2cm} x{2cm} } 
& \ballMid~(this work)& MidExtremes~\cite{fugger_et_al:LIPIcs.DISC.2018.27} & MH~\cite{10.1145/2488608.2488657} & VG~\cite{10.1145/2484239.2484256}\\
\hline
contraction rate & $\frac{1}{\sqrt{2}}$ & $\sqrt{\frac{7}{8}}$   & $\sqrt[d]{\frac{1}{2}}$ & $1-\frac{1}{n}$ \\
coordinate-free & yes & yes& no& yes\\
local time & $O(nd^{O(d)})$\tablefootnote{This result is due to~\cite{313559.313770}. There also exists a randomized algorithm with an expected runtime of $O(d^2n+2^{O(\sqrt{d\log d})})$~\cite{10.1145/201019.201036,10.1145/142675.142678}. For a fixed $d$, \ballMid is comparable to other choices of a point in the safe area.} & $O(n^2d)$ & $O(nd)$ & $O(nd)$ \\
\end{tabular}
\caption{Comparison of the results in this paper to prior work. Note that we added the hardness of computing the midpoint of the smallest enclosing ball for a more complete comparison. However, this computation step is not decisive compared to computing the safe area.}
\label{tab:comparison-prior-work}
\end{table}

\subsection{Technical Novelty}

In this work, we provide a new algorithm to compute a point in the convex hull of all correct nodes. This algorithm is based on the original algorithm by Mendes et al.~\cite{mendes2015multidimensional}, where the nodes locally compute a point inside the safe area. In our setting, nodes compute a point inside the intersection of convex hulls of subsets of $n-t$ received vectors. It has later been proven that choosing a point inside the safe area is necessary to satisfy convex validity~\cite{CAMBUS2026116044}. Thus, in our algorithm, the nodes locally compute the safe areas and choose the midpoint of the smallest enclosing ball of the safe area as the input for the next round. Here, the smallest enclosing ball is a fully-dimensional ball of smallest radius that still contains all vectors in a given set, such as the safe area.

The new algorithm idea does not directly imply the improved convergence rate. Additionally, we rely on the property that all safe areas must intersect. This contrasts with previous literature, where only a pairwise intersection of safe areas is used to show convergence. We show that, under synchronous communication, the safe areas must intersect in one point if the nodes use reliable broadcast for communication. To show that all safe areas intersect, we define a global view which contains input vectors for all nodes. This assumption is justified when reliable broadcast is used for communication, as at most one vector will be accepted for any corrupted node. We then show that the safe areas of all local views must contain the global safe area. This proof is based on the fact that each node locally maintains a subset of a global view, and computes the safe areas with potentially fewer sets. 

To show convergence, we adjust the definition of approximate agreement. Instead of arguing about convergence using the diameter of the vectors, we prove convergence using the radius of the smallest enclosing ball containing all non-corrupted vectors. We show that the radius of this ball shrinks by a factor of $\frac{1}{\sqrt{2}}$ in every iteration. This is done by carefully choosing a point in the intersection of all safe areas and upper bounding the distance from this point to the midpoints of the smallest enclosing balls of the safe areas.  

Under asynchronous communication, the previous argument used in the synchronous case would lead to reduced resilience. We thus rely on Gather~\cite{gather-original, gather-blog} to establish reliable communication between the nodes, such that all nodes accept the vectors of a common core of $n-t$ nodes in their local views. The common core may contain up to $t$ corrupted vectors. Using this communication primitive, we can define a global view based on the common core and show that each locally computed safe area will contain the global safe area as a subset.

Finally, we emphasize that previously used functions, such as MidExtremes, cannot reach the same contraction bound when using the assumption that the intersection of all safe areas is non-empty. Figure~\ref{fig:one-iteration-midextremes} shows an example where the contraction factor of MidExtremes is $\frac{\sqrt{3}}{2}\approx 0.866 > \frac{1}{\sqrt{2}}$ in one iteration, while the contraction factor of \ballMid is $\frac{\sqrt{3}}{3}\approx 0.577$ in the same example. 

\begin{figure}
    \centering
    \includegraphics[width=0.6\textwidth]{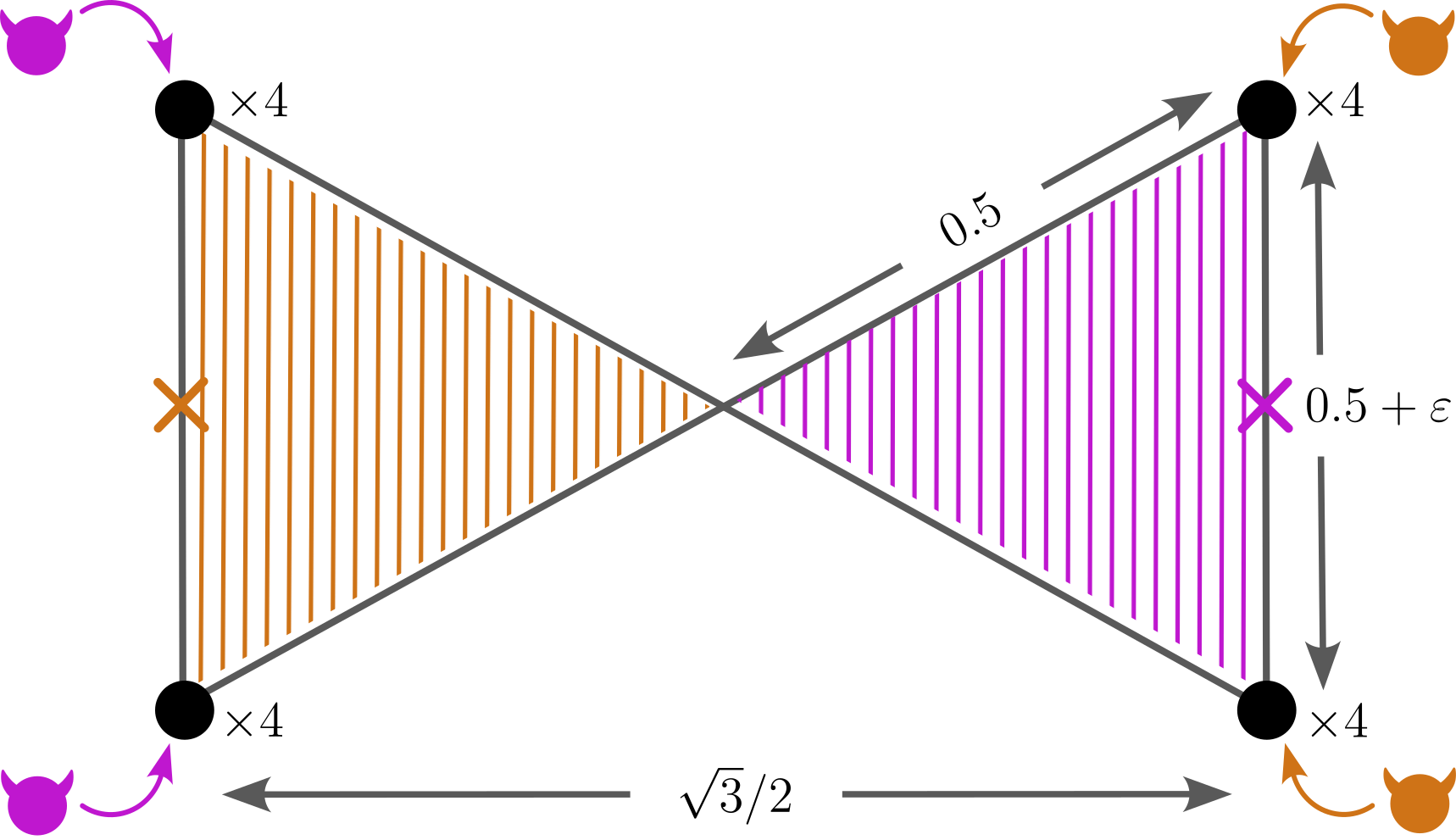}
    \caption{This is a construction showing that there can be iterations where MidExtremes has a larger contraction rate than $1/\sqrt{2}$, even if the safe areas computed by the nodes intersect. In this example, there are $20$ nodes, four of which are Byzantine. We assume that there are two possible views: in both views, all non-corrupted nodes accept all non-corrupted vectors. Additionally, in one view, the non-corrupted nodes accept the two Byzantine vectors on the left (pink nodes), while in the other view, they accept the two Byzantine vectors on the right (marked yellow). The corresponding safe areas are marked with stripes in the figure. The safe areas form triangles, where two sides are $0.5$ long, while the third side is $0.5+\varepsilon$, where $\varepsilon > 0$ can be chosen arbitrarily small. Note that the diameter of the non-corrupted input vectors as well as the diameter of their smallest enclosing ball are $1$. MidExtremes will choose the midpoint of the longest side as the new input value, marked with crosses in the figure. Thus, the inputs in the next iteration are at a distance $\sqrt{3}/2 > 1/\sqrt{2}$ away from each other. This construction works in the asynchronous communication model in our work. In the case of synchronous communication, one would have to replace the multiplicity of each correct vector to two instead of four nodes, and add a non-corrupted node at the intersection of the two safe areas which does not receive any Byzantine vectors at all.}
    \label{fig:one-iteration-midextremes}
\end{figure}

\section{Related Work}

Approximate agreement in one dimension was originally introduced by~\cite{10.1145/5925.5931} to overcome the FLP impossibility result~\cite{10.1145/3149.214121} of distributed computing, which states that no deterministic protocol in the asynchronous communication model that tolerates one crash failure can solve agreement. The idea of relaxing the agreement property has allowed to derive agreement algorithms that not only work in the asynchronous model, but also run in a logarithmic number of rounds in the size of the non-faulty input vectors.  

The multidimensional setting has been first considered by~\cite{mendes2015multidimensional}. The authors propose a solution for multidimensional Byzantine approximate agreement with convex validity. They show lower bounds on the resilience, which are $n>\max(3,d+1)f$ under synchronous communication, and $n>\max(3,d+2)f$ under asynchronous communication. Here, $f$ denotes the actual number of Byzantine nodes in the network. The contraction rate of the algorithms is  $1-1/n$~\cite{10.1145/2484239.2484256} and $\sqrt[d]{1/2}$~\cite{10.1145/2488608.2488657}. 
This contraction rate of multidimensional approximate agreement has been first improved in~\cite{fugger_et_al:LIPIcs.DISC.2018.27}, who considered the problem in dynamic networks. The authors provided two algorithms for the dynamic model that have a dimension-independent contraction rate. For one of these algorithms, the MidExtremes, the authors show that it can also be applied in a static fully-connected network, and has a contraction rate of $\sqrt{7/8}$. In~\cite{10.1145/3212734.3212762}, lower bounds for approximate agreement have been analyzed in the static and the dynamic settings. There, the $1/2$ bound is mentioned as the best-known lower bound for multidimensional Byzantine approximate agreement. Since its introduction, the multidimensional Byzantine approximate agreement problem has been extended to the network-agnostic setting~\cite{10.1145/3558481.3591105,10.1145/3796701.3815960}, to agreement on graphs~\cite{nowak_et_al:LIPIcs.DISC.2019.29, ALISTARH2023113733, 10.1145/3796701.3815967,rybicki2026solvabilityapproximateagreementgraphs}, and to weaker notions of output quality that do not require convex validity~\cite{CAMBUS2026116044,10.1145/3694906.3743343,cambus2025centroidapproximationbyzantinetolerantfederated, practicalvalidity}.

The smallest enclosing balls problem was first introduced in two dimensions in~\cite{seb-problem}, and later formally defined in $n$-dimensions in~\cite{https://doi.org/10.1002/nav.3800210414}. The textbook algorithm for this problem is Welzl's randomized algorithm~\cite{Welzl-SEB}, which has been extended and improved over the past years for practical applications~\cite{10.1007/978-3-540-39658-1_57,10.1145/996546.996548, doi:10.1137/070690419, 10.1145/142675.142678, 10.1145/201019.201036, 313559.313770}. In the context of agreement algorithms, the \seb has been used to define centroid approximation\cite{CAMBUS2026116044}, as well as a new validity condition~\cite{practicalvalidity} for multidimensional protocols. In the latter work, the authors propose agreement algorithms that choose a point inside the \seb of all non-faulty vectors. This validity condition relaxes the convex validity condition. Since only the server aggregates the vectors, no agreement algorithms are considered in that work.

\section{Model and Definitions}

We consider a fully connected network with $n$ nodes, identified by unique IDs $1,\ldots, n$. We assume that these IDs are known to all nodes in the network. The nodes communicate via authenticated channels. These channels are public (an eavesdropper can see the messages transmitted over the channel), reliable (a message sent by a node eventually arrives at the destination), and no cryptographic assumptions are made in this work. We assume that the message schedule is worst-case for a given protocol, i.e., an adversary controls the ordering in which the messages arrive. We further assume that the adversary may corrupt $f<t$ of the nodes in the network. Here, $f$ denotes the actual number of failures not known by the system, while $t$ is a known upper bound on the number of failures defined in the algorithm design. In this paper, we refer to the corrupted nodes as faulty or Byzantine and to the non-corrupted nodes as correct nodes.

\subsection{Communication models}
We consider the synchronous and the asynchronous communication models. In the synchronous communication model, the nodes wake up at the same time and start the protocol. The delay of each message sent over the channels is upper bounded by a known constant $\Delta$. We thus design round-based protocols in this model. In each synchronous round, a node can send a message, wait for messages from its neighbors, and perform local computation. In the asynchronous communication model, the message delay is unbounded. Thus, the communication is event-based. Also in the asynchronous communication model, it is possible to define rounds. However, these are defined locally by letting each node attach an iteration number to each message, identifying which iteration number of the algorithm the corresponding message belongs to. Since a Byzantine node, in particular, a crash node, is indistinguishable from a node communicating via links with large message delays, each node can wait for at most $n-f$ messages in a round before proceeding with the next steps. 

To establish more reliable communication between the nodes, and thus assimilate the nodes' local views, we will make use of additional communication primitives in both the synchronous and the asynchronous settings that we introduce next.

\paragraph*{Consistent broadcast for synchronous communication.} 
Under synchronous communication, all correct nodes receive all correct input values in the same round of communication. However, Byzantine nodes have the power to send different messages to different nodes. To avoid this, we use reliable broadcast~\cite{BrachaRB,srikanth1987simulating} to transmit messages in this case. Assume that the nodes broadcast their inputs as pairs $(i,v_i)$, where $i$ is the ID of the node, and $v_i$ is the corresponding input value. One execution of consistent broadcast\cite{cachin2011introduction} as presented in Algorithm~\ref{alg:consistentBC} provides the following properties while tolerating up to $t<n/3$ Byzantine nodes:

\begin{itemize}
    \item \textbf{Validity:} If a correct node $i$ broadcasts a message $(i,v_i)$, then every correct process reliably accepts this message eventually. 
    \item \textbf{No duplication:} Every correct node reliably accepts at most one message from each sender.
    \item \textbf{Consistency:} If two correct nodes reliably accept the messages $(i,v_i)$ and $(i,v'_i)$, then, $v_i=v'_i$.
    \item \textbf{Integrity:} If some correct node reliably accepts a message $v_i$ with sender $i$ and node $i$ is correct, then $v_i$ was previously broadcast by $i$.
\end{itemize}

\begin{algorithm}[tbh]
\caption{Synchronous consistent broadcast for $n>3t$}
\label{alg:consistentBC}
\begin{algorithmic}[1]
    \State Each node $i \in [n]$ with input value $v_i$ executes the following code:
    \Indent
        \State Broadcast $(i,v_i)$
        \State Store all received $(j,v_j)$ in the set $S_j$
        \State Broadcast $S_i$
        \State Accept $(j,v_j)$ if $(j,v_j)$ was included in at least $n-t$ many sets $S_k$
    \EndIndent
\end{algorithmic}
\end{algorithm}

\paragraph*{Gather primitive for asynchronous communication.}
In the asynchronous communication model, we use a stronger broadcast primitive to exchange the messages, called Gather. The original protocol is due to Canetti and Rabin~\cite{gather-original}. Algorithm~\ref{alg:gather} presents the Gather primitive as in~\cite{gather-blog}. The Gather protocol fulfills the following three properties and can tolerate up to $t<n/3$ Byzantine nodes:
\begin{itemize}
    \item \textbf{Common core:} There exists a set $S$ of input values of size $n-t$  that is included in the output sets of all correct nodes.
    \item \textbf{Validity:} If a correct node $i$ includes a message $(j,v_j)$ from another correct node $j$ in its output set, then $v_j$ is $j$'s input.
    \item \textbf{Consistency:} If two correct nodes respectively include the messages $(i,v_i)$ and $(i,v'_i)$ in their output sets, then $v_i=v'_i$.
\end{itemize}

Other than in the synchronous case, we cannot make sure that the common core of $n-t$ values will include correct values only. In the worst case, all $t$ values may come from Byzantine nodes.

\begin{algorithm}[tbh]
\caption{Gather for $n>3t$}
\label{alg:gather}
\begin{algorithmic}[1]
    \State Each node $i \in [n]$ with input value $v_i$ executes the following code:
    \Indent
        \State \textbf{Broadcast} $(i,v_i)$ reliably.\Comment{Using Bracha's reliable broadcast}
        \State \textbf{Wait} until $n-t$ messages $(j,v_j)$ have been reliably received and store all accepted messages inside $S_i$
        \State \textbf{Broadcast} $S_i$
        \State \textbf{Wait} until $n-t$ sets $S_j$ have been accepted. A set $S_j$ is considered accepted if $i$ has reliably accepted every pair $(k,v_k)\in S_j$. Set $T_i = \bigcup S_j$, where $S_j$ is an accepted set
        \State \textbf{Broadcast} $T_i$
        \State \textbf{Wait} until $n-t$ sets $T_j$ are accepted. A set $T_j$ is considered accepted if $i$ has reliably accepted every pair $(k,v_k)\in T_j$. Set $U_i = \bigcup T_j$, where $T_j$ is an accepted set
        \State\Return $U_i$
    \EndIndent
\end{algorithmic}
\end{algorithm}

\subsection{Multidimensional Byzantine approximate agreement}

In multidimensional Byzantine approximate agreement, the input values of the nodes are vectors $v\in V$. In this work, we assume $V$ to be a finite-dimensional real vector space equipped with an inner product $\langle\cdot,\cdot\rangle: V\times V\rightarrow \mathbb{R}$ and the norm $\lVert x\rVert = \sqrt{\langle x,x\rangle}$. We use $d$ to denote the dimension of $V$. In practical applications, $V$ is often assumed to be $\mathbb{R}^d$, and the distances are measured in terms of the Euclidean norm. Before providing a formal definition of the problem, we first define the convex hull and a closed ball: 

\begin{definition}[Convex set]
    Let $V$ be a real vector space. A subset $W\subseteq V$ is called a \emph{convex set} if for any two elements $x,y\in W$ and all $\lambda\in[0,1]$ holds
    $$\lambda x + (1-\lambda)y \in W.$$
\end{definition}

\begin{definition}[Convex hull]
    A \emph{convex hull} of a set $U\subseteq V$, denoted  $\convexHull(U)$ is the smallest convex set $W$ in $V$ such that $U\subseteq W$.  
\end{definition}

\begin{definition}[Closed ball]
    Let $m\in V$ and $r\in\mathbb{R}, r\ge 0$. A \emph{closed ball} with midpoint in $m$ and radius $r$ is defined as 
    $$\ballOf{m,r} = \{x\in V: \lVert x - m\rVert \le r\}.$$
\end{definition}

Using these definitions, we can now formally define the multidimensional Byzantine approximate agreement problem. 

\begin{definition}[MBAA, adapted version from~\cite{mendes2015multidimensional}]\label{def:MBAA}
    A protocol solves the \emph{Multidimensional Byzantine Approximate Agreement} problem (MBAA) if it satisfies
    \begin{itemize}
        \item $\varepsilon$-\textbf{Agreement:} The output vectors of correct nodes should lie inside a $d$-dimensional closed ball of radius $\varepsilon$.
        \item \textbf{Containment:} The output of all correct nodes must be inside the convex hull of the correct input vectors.
        \item \textbf{Termination:} Each correct node must terminate within a finite amount of time.
    \end{itemize}
\end{definition}
Note that we adapted the $\varepsilon$-Agreement definition in this work (compared to~\cite{mendes2015multidimensional}), such that the nodes must output vectors inside an $\varepsilon$-ball, instead of considering the diameter of the points. We believe a ball-based definition is cleaner for our work, since our algorithms and analysis rely on closed balls. Note that the diameter of the points in our case may be up to $2\varepsilon$ at the time of termination. In contrast, traditional algorithms that terminate with maximum diameter $\varepsilon$ may lie within a ball of radius $r\le\varepsilon\sqrt{\frac{d}{2(d+1)}}$. This bound is given by Jung's theorem~\cite{jung1901kleinste} and is tight if the vectors form a regular $d$-simplex. Thus, the definition considered here is equivalent to the definition in~\cite{mendes2015multidimensional} up to a constant factor $\leq 2$. Moreover, the algorithms in this paper satisfy the original definition of MBAA as in~\cite{mendes2015multidimensional} if two additional iterations are executed. With a contraction rate of $\frac{1}{\sqrt{2}}$ per iteration, this ensures that all correct input vectors are in a ball of radius $\varepsilon/2$, which implies a maximum diameter of $\varepsilon$.

\section{Synchronous multidimensional approximate agreement}\label{sec:synch-MBAA}

In this section, we present the synchronous MBAA algorithm. In this algorithm, the nodes exchange their input vectors reliably in every iteration. Based on the received vectors, they locally compute an area where the convex validity is satisfied. As in~\cite{mendes2015multidimensional}, this is done by computing a so-called safe area on the set of received vectors:

\begin{definition}[Safe area]
    Let $\{v_1,\ldots,v_n\}\subset V$ be a set of $n$ vectors. Given $\ell<n$, the \emph{safe area} on $\ell$ nodes is computed as
    $$\safe_\ell = \bigcap_{\substack{I\subseteq \left[n\right]\\ |I|=\ell}}\convexHull\bigl(\{v_i, i\in I\}\bigr).$$
\end{definition}

Note that the safe area in this definition may be empty. Therefore, the safe area in the algorithm is computed on intersections of $\ell=n-t$ vectors, which guarantees that the safe area is non-empty. 

In order to choose a point inside the safe area, we use the midpoint of the smallest enclosing ball of the safe area, which is defined as follows:

\begin{definition}[Enclosing ball]
    Let $U\subseteq V$. A closed ball $\ballOf{m,r}$ is an \emph{enclosing ball} of the set $U$ if $U\subseteq \ballOf{m,r}$.
\end{definition}

\begin{definition}[Smallest enclosing ball]
    The \emph{smallest enclosing ball} of a set $U\subseteq V$, denoted $\ballOf{U}$, is an enclosing ball $\ballOf{m^*, r^*}$ of $U$ such that $r^*\le r$ holds for every enclosing ball $\ballOf{m, r}$ of $U$.
\end{definition}

The smallest enclosing ball has been shown to be unique. Moreover, in $d$ dimensions, the smallest enclosing ball is always defined by at most $d+1$ points. 
We use $\radius(\ballOf{U})$ to denote the radius of the smallest enclosing ball of the set $U$.

\begin{definition}[Ball midpoint]
    Let $U\subseteq V$. The \emph{ball midpoint} of $U$, denoted \ballMidOf{U}, is the midpoint of $\ballOf{U}$.
\end{definition}

We further define $C$ to be the set of correct input vectors, and $\ballOf{C}$ to denote the smallest enclosing ball of all correct input vectors. For the following sections, we assume that an upper bound $c\cdot\radius(\ballOf{C})$ on the \seb of all correct input vectors is known to all nodes, where $c\ge 1$ is some small constant. In Section~\ref{sec:round-estimation}, we discuss how this assumption can be dropped, and how an upper bound on the number of iterations can be computed locally by each node. Algorithm~\ref{alg:synch-ballmidpoint} presents the synchronous algorithm solving MBAA as pseudocode.

\begin{algorithm}[tbh]
\caption{Synchronous MBAA with \ballMid for $n>\max(3,d+1)t$}
\label{alg:synch-ballmidpoint}
\begin{algorithmic}[1]
    \State Each node $i \in [n]$ with input vector $v_i^0$ executes the following code:
    \Indent
        \For{ $r=1,2,\ldots,\log_2 \bigl(\frac{1}{\varepsilon}\cdot c\cdot\radius(\ballOf{C})\bigr)/\log_2(\sqrt{2})$}
            \State Consistently broadcast current value $v_i^{r-1}$
            \State Store accepted values $v_j^{r-1}$ in $V_i^{r-1} = \{v_j^{r-1}, j\in [n]\}$
            \State $v_i^r \leftarrow \ballMidOf{\safe_{n-t}(V_i^{r-1})}$\label{line:synch-ball}
        \EndFor
        \State\Return $v_i^r$
    \EndIndent
\end{algorithmic}
\end{algorithm}

\subsection{Analysis of Algorithm~\ref{alg:synch-ballmidpoint}}

In this section, we analyze the correctness and the contraction rates of Algorithm~\ref{alg:synch-ballmidpoint}. In particular, we prove the following theorem:

\begin{theorem}\label{thm:synch-contraction}
    Algorithm~\ref{alg:synch-ballmidpoint} solves MBAA according to Definition~\ref{def:MBAA} in the synchronous communication setting in $\log_2 \bigl(\frac{1}{\varepsilon}\cdot c\cdot\radius(\ballOf{C})\bigr)$ iterations when $n>\max(3,d+1)t$, where $C$ is the set of correct input vectors. The contraction rate in every iteration is upper bounded by $\frac{1}{\sqrt{2}}$.
\end{theorem}

We start by proving the correctness of the algorithm. We first make use of the results from Mendes et al.~\cite{mendes2015multidimensional} stating that a safe area computed on subsets of $n-t$ vectors is not empty when $n>\max(3,d+1)t$, and that it is included in the convex hull of all correct vectors. Let $C^{r-1}$ denote the set of correct vectors $v_i^{r-1}$ broadcast in iteration $r-1$, where $C^{0} = C$. Let $Corr$ denote the set of IDs of the correct nodes.

\begin{lemma}\label{lem:synch-safearea-nonempty}
    The safe areas $\safe_{n-t}(V_i^{r-1})$ computed by each node in Line~\ref{line:synch-ball} of Algorithm~\ref{alg:synch-ballmidpoint} are non-empty for $n>\max(3,d+1)t$. Additionally, $\safe_{n-t}(V_i^{r-1})\subseteq\convexHull(C^{r-1})\ \forall i\in Corr.$
\end{lemma}

The proof of this lemma follows directly from the work in \cite{mendes2015multidimensional}. For the non-empty intersection, the authors show that due to the inequality $n>\max(3,d+1)t$, any $d+1$ convex hulls of $n-t$ vectors must intersect. They then apply Helly's theorem~\cite{MR157289} to deduce that all such subsets must have a non-empty intersection. Note that in the original proof, $f$ instead of $t$ was used in the argument, and also the resilience bound was stated with $f$. The property that the safe area is included in the convex hull of all correct vectors follows from the fact that at least one of the subsets of $n-t$ vectors that are used locally to compute the safe area consists only of correct nodes. Thus, any intersection of the convex hulls will be inside the convex hull of these $n-t$ correct nodes.

In the next step, we show that the midpoint of the ball is inside the safe area:

\begin{lemma}\label{lem:synch-ballmidpoint-in-safearea}
    The ball midpoint $\ballMidOf{\safe_{n-t}(V_i^{r-1})}$ computed by each node in Line~\ref{line:synch-ball} of Algorithm~\ref{alg:synch-ballmidpoint} is inside the safe area $\safe_{n-t}(V_i^{r-1})$ of the respective node. 
\end{lemma}
\begin{proof}
    Let $m_i\coloneqq\ballMidOf{\safe_{n-t}(V_i^{r-1})}$. Assume by means of contradiction that $m_i\notin \safe_{n-t}(V_i^{r-1})$. Note that the safe area by definition forms a convex set, as it is the intersection of convex sets. Since $m_i$ is not included in the safe area, there exists a hyperplane separating $m_i$ from $\safe_{n-t}(V_i^{r-1})$ that does not include $m_i$. This hyperplane divides the $\ballOf{\safe_{n-t}(V_i^{r-1})}$ in two spherical caps. We now define $B^*$ to be the \seb of the intersection of this hyperplane with the $\ballOf{\safe_{n-t}(V_i^{r-1})}$. Note that $B^*$ has a strictly smaller radius than $\ballOf{\safe_{n-t}(V_i^{r-1})}$, as it does not include $m_i$. Moreover $B^*$ includes the minor spherical cap of $\ballOf{\safe_{n-t}(V_i^{r-1})}$, and thus the $\safe_{n-t}(V_i^{r-1})$. This is however a contradiction to $\ballOf{\safe_{n-t}(V_i^{r-1})}$ being the \seb of the safe area, and thus the theorem statement holds.

\end{proof}

Note that in \cite{mendes2015multidimensional} it was shown that any algorithm that deterministically chooses a point inside the safe area solves MBAA. The above lemma thus concludes the convergence discussion of our algorithm. It thus remains to show that the contraction rate and the resulting number of iterations to reach the $\varepsilon$-Agreement are as claimed.

In the following part, we focus on proving the contraction rate and thus the runtime of the algorithm.

\begin{lemma}\label{lem:safe-in-trueball}
    All locally computed safe areas $\safe_{n-t}(V_i^{r-1}), i\in Corr$ by the correct nodes in an iteration of Algorithm~\ref{alg:synch-ballmidpoint} are inside the \seb of all correct vectors in that iteration. That is, $\safe_{n-t}(V_i^{r-1})\subseteq\ballOf{C^{r-1}}\ \forall i\in Corr.$
\end{lemma}
\begin{proof}
    Observe first that $\safe_{n-t}(V_i^{r-1})\subseteq\convexHull(C^{r-1})\ \forall i\in Corr$ by Lemma~\ref{lem:synch-safearea-nonempty}. By definition of the \seb, we have $\convexHull(C^{r-1})\subseteq\ballOf{C^{r-1}}$, which concludes the proof.
\end{proof}

In the next lemma, we show that $\ballOf{C^{r-1}}$ includes a hemisphere of each locally computed ball $\ballOf{\safe_{n-t}(V_i^{r-1})}, i\in Corr$. Note that the intersection of these two balls is already guaranteed by Lemma~\ref{lem:safe-in-trueball}. 

\begin{lemma}\label{lem:synch-diameter-included}
    Let $m$ be the midpoint of $\ballOf{C^{r-1}}$, and $m_i\coloneq\ballMidOf{\safe_{n-t}(V_i^{r-1})}$ denote the midpoint of the \seb of the locally computed safe area of node $i$. Let $\overrightarrow{m,m_i}$ be the line going through the midpoints of the two balls. $\ballOf{C^{r-1}}$ includes the diameters of $\ballOf{\safe_{n-t}(V_i^{r-1})}$ that are orthogonal to $\overrightarrow{m,m_i}$.
\end{lemma}
\begin{proof}
    We prove the lemma for a fixed correct node $i$ and a fixed iteration $r-1$. Observe first that the midpoint of $\ballOf{\safe_{n-t}(V_i^{r-1})}$ is included in $\ballOf{C^{r-1}}$, because of the combination of Lemma~\ref{lem:synch-ballmidpoint-in-safearea} and Lemma~\ref{lem:safe-in-trueball}.

    Note next that, if a diameter of $\ballOf{\safe_{n-t}(V_i^{r-1})}$ that is orthogonal to $\overrightarrow{m,m_i}$ is included in $\ballOf{C^{r-1}}$, then all orthogonal diameters are included in $\ballOf{C^{r-1}}$, and thus a hemisphere of $\ballOf{\safe_{n-t}(V_i^{r-1})}$ is inside $\ballOf{C^{r-1}}$. This is because any hyperplane in a $d$-dimensional vector space orthogonal to a line through $m$ intersects a $d$-dimensional ball $\ballOf{C^{r-1}}$ in either a $d-1$-dimensional ball, or a single point, or not at all. Moreover, the midpoint of the $d-1$-dimensional ball representing the intersection ball will lie on the line through $m$. Since the midpoint of $\ballOf{\safe_{n-t}(V_i^{r-1})}$ is on the line $\overrightarrow{m,m_i}$, and a diameter of the ball which is orthogonal to $\overrightarrow{m,m_i}$ is included in $\ballOf{C^{r-1}}$, all diameters must be included as well.

    We first cover two special cases: if the intersection of the two balls is a single point, then $\ballOf{\safe_{n-t}(V_i^{r-1})}$ degenerates to a single point, and the lemma statement holds. The lemma also holds trivially if $\ballOf{\safe_{n-t}(V_i^{r-1})}\subseteq\ballOf{C^{r-1}}$. We exclude these two cases for the following analysis.

    Assume by means of contradiction that none of the diameters of $\ballOf{\safe_{n-t}(V_i^{r-1})}$ that are orthogonal to $\overrightarrow{m,m_i}$ are included in $\ballOf{C^{r-1}}$, i.e., the hyperplane defined by the orthogonal diameters intersects the boundary of $\ballOf{\safe_{n-t}(V_i^{r-1})}$ outside of $\ballOf{C^{r-1}}$. Note that all vectors that define $\ballOf{\safe_{n-t}(V_i^{r-1})}$, which are in the safe area, are inside $\ballOf{C^{r-1}}$ (see Lemma~\ref{lem:safe-in-trueball}). In particular, these vectors are inside the intersection $\ballOf{\safe_{n-t}(V_i^{r-1})}\cap\ballOf{C^{r-1}}$. Consider the intersection $\ballOf{\safe_{n-t}(V_i^{r-1})}\cap\ballOf{C^{r-1}}$. Note that the diameter of this intersection is strictly smaller than the diameter of $\ballOf{\safe_{n-t}(V_i^{r-1})}$, since the balls are not identical. Therefore, the $d$-dimensional \seb of this intersection has a smaller radius than $\ballOf{\safe_{n-t}(V_i^{r-1})}$. Moreover, this ball includes $\safe_{n-t}(V_i^{r-1})$, meaning that we found an enclosing ball of $\safe_{n-t}(V_i^{r-1})$ which is smaller than the smallest enclosing ball $\ballOf{\safe_{n-t}(V_i^{r-1})}$, which is a contradiction. This shows that at least one diameter of $\ballOf{\safe_{n-t}(V_i^{r-1})}$ that is orthogonal to $\overrightarrow{m,m_i}$ must be inside $\ballOf{C^{r-1}}$ and thus also a whole hemisphere of $\ballOf{\safe_{n-t}(V_i^{r-1})}$.

\end{proof}

We now show one of the main lemmas needed to prove contraction. In this lemma, we show that the safe areas computed by the correct nodes in the same iteration have a non-empty intersection. 

\begin{lemma}\label{lem:synch-all-safearea-intersect}
    Let $\safe_{n-t}(V_i^{r-1}), i\in Corr$ be the locally computed safe areas by correct nodes in iteration $r-1$ of Algorithm~\ref{alg:synch-ballmidpoint}. Then, 
    $$\bigcap_{i\in Corr}\safe_{n-t}(V_i^{r-1})\neq \varnothing.$$
\end{lemma}
\begin{proof}
    Consider iteration $r-1$ of Algorithm~\ref{alg:synch-ballmidpoint}. Note that, under consistent broadcast and in particular its implementation as in Algorithm~\ref{alg:consistentBC}, all correct nodes accept all correct vectors broadcast in this iteration. This reflects the validity property of consistent broadcast. Now consider a global view of the system, where $n-f$ correct and $f$ Byzantine nodes are present. Due to the no duplication property of consistent broadcast, each Byzantine party can have at most one vector accepted by the correct nodes in one consistent broadcast iteration. 
    
    Using these properties, we will now define a global view. This view corresponds to all accepted vectors (one per sender) in an iteration of the algorithm. Note that there are exactly $n-f$ correct accepted vectors and there can be at most $f$ different Byzantine vectors in this view. We will assume here there are exactly $f$ different Byzantine vectors, and add more deterministically chosen vectors if fewer than $f$ vectors have been broadcast by Byzantine nodes. Let $G^{r-1}$ denote this global set of vectors in iteration $r-1$. Consider now $\safe_{n-t}(G^{r-1})$, i.e., a globally computed safe area under the assumption that $f$ is not known. By Lemma~\ref{lem:synch-safearea-nonempty}, $\safe_{n-t}(G^{r-1})\neq \varnothing$, as $G$ may be the local view of one of the correct nodes. We will show that this safe area is a subset of the local safe area of each correct node, that is $\safe_{n-t}(G^{r-1})\subseteq\safe_{n-t}(V_i^{r-1}), \forall i\in Corr$.

    The main observation is that every correct node $i$ will accept a subset of the vectors in $G^{r-1}$, i.e., $V_i^{r-1} \subseteq G^{r-1}\ \forall i\in Corr$. This subset includes the $n-f$ correct vectors and potentially fewer Byzantine vectors than in the global view. Consider the convex hulls of subsets of $n-t$ vectors in $G^{r-1}$ which are used to compute the global safe area. Note that by the above property, the local safe area $\safe_{n-t}(V_i^{r-1})$ of a correct node $i$ is computed by intersecting a subset of convex hulls that are used for the computation of the global safe area. Thus, the locally computed safe area always includes the global safe area $\safe_{n-t}(G^{r-1})$.
\end{proof}

\begin{lemma}\label{lem:synch-contraction-rate}
    The contraction rate of Algorithm~\ref{alg:synch-ballmidpoint} in every iteration is upper bounded by $\frac{1}{\sqrt{2}}$.
\end{lemma}
\begin{proof}

Consider iteration $r-1$ of the algorithm. Assume for this proof that the midpoint of $\ballOf{C^{r-1}}$ is at the origin $O$. Note that this can be achieved by translation of the construction, without influencing the distances. Let $m_i$ be the midpoints of $\ballOf{\safe_{n-t}(V_i^{r-1})}$ where $i\in Corr$. Let $\rho_i$ be the respective radii $\rho_i\coloneqq\radius(\ballOf{\safe_{n-t}(V_i^{r-1})})$ and let $q$ denote some point inside the intersection of the smallest enclosing balls of the safe areas of all correct nodes, i.e., $q\in\bigcap_{i\in Corr}\ballOf{\safe_{n-t}(V_i^{r-1})}$. Such a point exists by Lemma~\ref{lem:synch-all-safearea-intersect}. 

Note first that, by definition of $q$, the following inequality holds:
\begin{align}\label{eq:distance-to-intersection-point}
   \lVert m_i - q\rVert^2\le \rho_i^2 
\end{align}
This is because $q$ is included in every \seb $\ballOf{\safe_{n-t}(V_i^{r-1})}$ with midpoint in $m_i$.

By Lemma~\ref{lem:synch-diameter-included}, we know that all diameters of the locally computed balls $\ballOf{\safe_{n-t}(V_i^{r-1})}$ that are orthogonal to $\overrightarrow{O,m_i}$ are included in $\ballOf{C^{r-1}}$. Let $\Rho\coloneqq\radius(\ballOf{C^{r-1}})$ be the radius of the true ball. Then, 
\begin{align}\label{eq:kinda-pythagoras}
    \lVert m_i - O\rVert^2 + \rho_i^2 \le \Rho^2
\end{align}
This inequality follows by applying the theorem of Pythagoras. Assume that the diameter of $\ballOf{\safe_{n-t}(V_i^{r-1})}$ is orthogonal to $\overrightarrow{O,m_i}$ is a chord of $\ballOf{C^{r-1}}$. In this case, the above equation would be satisfied with equality. Since the diameter may be strictly inside $\ballOf{C^{r-1}}$, the distance from the origin to one of the endpoints of the diameter can be upper bounded by the radius of $\ballOf{C^{r-1}}$, giving the above inequality.

By adding the inequalities~(\ref{eq:distance-to-intersection-point}) and~(\ref{eq:kinda-pythagoras}), we get 
$$\lVert m_i-q\rVert^2 +\lVert m_i\rVert^2 +\rho_i^2 \le \Rho^2 + \rho_i^2.$$
That is, 
\begin{align}\label{eq:sum-of-inequalities}
    \lVert m_i-q\rVert^2 +\lVert m_i\rVert^2 \le \Rho^2.
\end{align}
We can reformulate the left side as
\begin{align*}
    \lVert m_i-q\rVert^2 +\lVert m_i\rVert^2 &= \lVert m_i-q/2 - q/2\rVert^2 +\lVert m_i-q/2 + q/2\rVert^2\\
    &=2\lVert m_i-q/2\rVert^2 + 2\lVert q/2\rVert^2 - 2\langle m_i - q/2,q/2\rangle + 2\langle m_i - q/2,q/2\rangle \\
    &= 2\lVert m_i-q/2\rVert^2 + \lVert q\rVert^2/2
\end{align*}
and rewrite inequality~(\ref{eq:sum-of-inequalities}) as

$$2\lVert m_i-q/2\rVert^2 +\lVert q\rVert^2/2 \le \Rho^2$$
$$\lVert m_i-q/2\rVert \le \sqrt{\frac{2\Rho^2 - \lVert q\rVert^2}{4}}.$$

The right side of the inequality is maximized for $\lVert q\rVert = 0$. We thus get 

$$\lVert m_i-q/2\rVert \le \frac{1}{\sqrt{2}}\Rho$$

Note that this provides an upper bound on the distance from any midpoint to the point $q/2$. That is, there is a ball of radius at most $\frac{1}{\sqrt{2}}\Rho$ around $q/2$ that includes all $v_i^r =\ballMidOf{\safe_{n-t}( V_i^{r-1})}$, which are the input vectors of the correct nodes for the next algorithm iteration. 

This shows a contraction factor of $\frac{1}{\sqrt{2}}$ of the \ballMid algorithm under synchronous communication.
\end{proof}

Lemma~\ref{lem:synch-contraction-rate} showed that the radius of the \seb of all correct vectors shrinks by at least a factor of $\frac{1}{\sqrt{2}}$  in each iteration. For the radius to reach below $\varepsilon$, $\log_{\sqrt{2}}(\frac{1}{\varepsilon}\cdot\radius(\ballOf{C}))$ iterations of the algorithm are required. This leads to the following corollary:

\begin{corollary}
    Algorithm~\ref{alg:synch-ballmidpoint} converges in $O\left(\log_2 \bigl(\frac{1}{\varepsilon}\cdot c\cdot\radius(\ballOf{C})\bigr)\right)$ such that the output vectors satisfy $\varepsilon$-Agreement.
\end{corollary}

\subsection{Tightness of the contraction bound for Algorithm~\ref{alg:synch-ballmidpoint}}\label{sec:tightness}

In this section, we show that the contraction rate is asymptotically tight for Algorithm~\ref{alg:synch-ballmidpoint}. To this end, we present a construction and a Byzantine attack which shows that the contraction rate of the algorithm can be arbitrarily close to $\frac{1}{\sqrt{2}}$ in every iteration of the algorithm.

\begin{figure}[h]
    \centering
    \includegraphics[width=\textwidth]{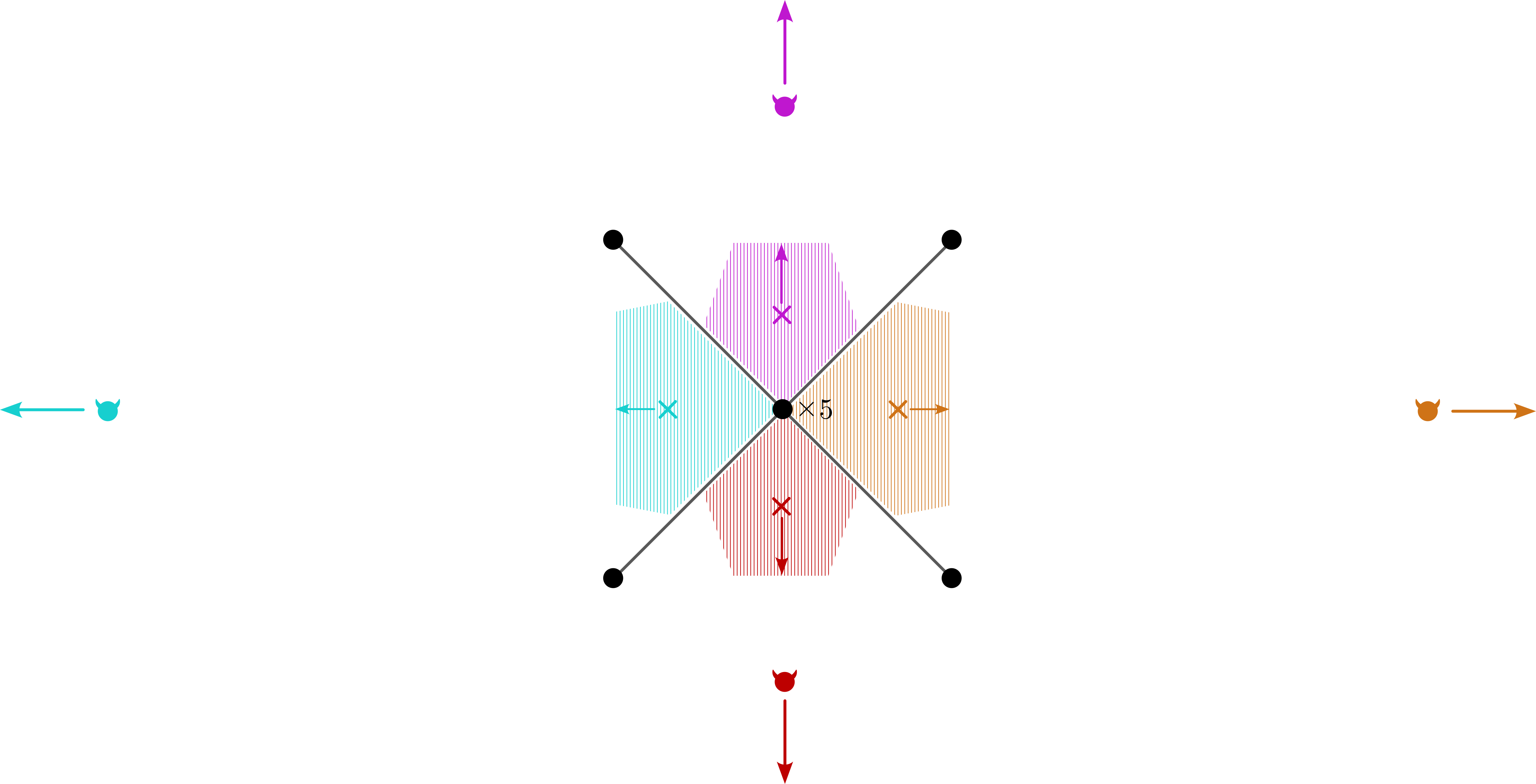}
    \caption{There are $9$ correct and $4$ Byzantine nodes in this construction. $4$ correct vectors form a square, with $5$ correct nodes in the middle, and all correct nodes will accept these vectors. Additionally, each correct node at the vertex of the square accepts one of the four Byzantine vectors. The Byzantine vectors are assumed to be placed far away from the correct vectors in the marked direction. The further the placement of the Byzantine vectors, the closer the next input vector of the correct nodes to the midpoint of a side of the square. The safe areas corresponding to the five different views, one of which is just the origin, are marked in the figure. Note that new input vectors of the correct nodes for the next round are also forming a square from points that are close to the midpoints of the sides. Thus, Byzantine nodes can apply an analogous attack in the next iteration. }
    \label{fig:lower-bound-midpoint-synch}
\end{figure}

\begin{theorem}\label{thm:lower-bound-synch}
    The contraction rate of Algorithm~\ref{alg:synch-ballmidpoint} is at least $\frac{1}{\sqrt{2}}-\delta$, where $\delta$ is an arbitrarily small constant.
\end{theorem}
\begin{proof}
    Consider a construction in $\mathbb{R}^2$ with $f\ge4$ Byzantine and $n>3f$ correct nodes. For larger $d$, the same example can be used where all but two dimensions are $0$. Place input vectors of four correct nodes in the vertices of a square, with coordinates $\left(\frac{1}{\sqrt{2}},\frac{1}{\sqrt{2}}\right), \left(-\frac{1}{\sqrt{2}},\frac{1}{\sqrt{2}}\right), \left(\frac{1}{\sqrt{2}},-\frac{1}{\sqrt{2}}\right)$ and $\left(-\frac{1}{\sqrt{2}},-\frac{1}{\sqrt{2}}\right)$. Place the remaining $n-f-4$ correct nodes at the origin. Note that the \seb of the square is $1$ in this construction. 
    
    We now describe the adversarial strategy. Let four Byzantine nodes choose their inputs to be $(\alpha,0)$, $(-\alpha,0)$, $(0,\alpha)$ and $(0,-\alpha)$, where $\alpha>2$ is a large constant which we will define later. The remaining Byzantine nodes do not broadcast their own input vectors, but may participate in the broadcast of other vectors. Since the Byzantine adversary controls the scheduling of the messages, it is possible for the adversary to make sure that four correct nodes each accept one different input vector from a Byzantine node. The remaining $n-f-4$ correct nodes do not accept any Byzantine vectors.

    The four correct nodes that receive Byzantine inputs receive $n-f+1$ values in total and therefore remove one vector to compute the safe areas. The four different safe areas all intersect in the origin. We will refer to these safe areas as the main safe areas. The areas are bounded by the diagonals of the square, and, for $\alpha \rightarrow \infty$, they converge to triangles bounded by one side and the two diagonals of the square. The remaining group of $n-f-4$ correct nodes has a safe area that contains only the origin. See Figure~\ref{fig:lower-bound-midpoint-synch} for an example with $13$ nodes. 
    Observe that for $\alpha > 2$, the midpoints of the \seb of the main safe areas are $\left(\frac{\alpha}{\sqrt{2}\alpha + 2}, 0 \right)$, $\left(-\frac{\alpha}{\sqrt{2}\alpha + 2}, 0 \right)$, $\left(0, \frac{\alpha}{\sqrt{2}\alpha + 2}\right)$ and $\left(0, -\frac{\alpha}{\sqrt{2}\alpha + 2} \right)$. 

    For the radius of the new \seb of correct vectors, we have 
    \begin{align*}
        \frac{\alpha}{\sqrt{2}\alpha +2} = \frac{1}{\sqrt{2}}\cdot\frac{1}{1+\frac{\sqrt{2}}{\alpha}} = \frac{1}{\sqrt{2}}\cdot\left(1-\frac{\sqrt{2}}{\alpha} + O\left(\frac{1}{\alpha^2}\right)\right) = \frac{1}{\sqrt{2}} - \frac{1}{\alpha} + O\left(\frac{1}{\alpha^2}\right).
    \end{align*}
    Note that we used the Taylor series to expand the second term in the second equality. 
    By choosing $\alpha > \frac{1}{\delta}$, we get that the radius of the \seb after one algorithm iteration is at least $\frac{1}{\sqrt{2}}-\delta$.

    Observe that the inputs for the correct nodes in the next iteration are geometrically the same as in the initial construction up to rotation and scaling (corresponding to the contraction factor). Thus, we can repeat the same attack where the Byzantine nodes rotate their input vectors to $(\alpha,\alpha)$, $(-\alpha,\alpha)$, $(\alpha,-\alpha)$, and $(-\alpha,-\alpha)$. By iteratively repeating this process, we receive an instance where the contraction factor of Algorithm~\ref{alg:synch-ballmidpoint} is at least $\frac{1}{\sqrt{2}}-\delta$ in every iteration.
\end{proof}

\section{Asynchronous multidimensional approximate agreement}\label{sec:asynch-MBAA}

In this section, we present the asynchronous algorithm solving MBAA with a contraction rate of $\frac{1}{\sqrt{2}}$. The procedure is similar to that in the synchronous case, and we will therefore reuse many of the previously stated results. The first difference is that we use the Gather protocol to communicate the vectors in each iteration, as consistent broadcast is not sufficient to guarantee nonempty intersection of all local safe areas. The second difference is the computation of the safe area. In the synchronous case, all nodes receive at least $n-t$ correct vectors. In the asynchronous case, we can only guarantee that the nodes accept a common core of $n-t$ vectors. Additionally, depending on the exact number of accepted vectors, a node can compute a safe area using from $n-2t$ vectors (if it accepts $n-t$ vectors) up to $n-t$ vectors (if it accepts $n$ vectors). That is, the safe area can be computed by removing subsets of $t$ nodes rather than using subsets of $n-2t$ nodes. Algorithm~\ref{alg:asynch-ballmidpoint} presents this idea as pseudocode. As in the synchronous case, we first assume that an upper bound on the \seb of all correct input vectors is known to all nodes, and present the round estimation preprocessing step in Section~\ref{sec:round-estimation}.

\begin{algorithm}[tbh]
\caption{Asynchronous MBAA with \ballMid for $n>\max(3,d+2)t$}
\label{alg:asynch-ballmidpoint}
\begin{algorithmic}[1]
    \State Each node $i \in [n]$ with input vector $v_i^0$ executes the following code:
    \Indent
        \For{ $r=1,2,\ldots,\log_2 \bigl(\frac{1}{\varepsilon}\cdot c\cdot\radius(\ballOf{C})\bigr)/\log_2(\sqrt{2})$}
            \State Execute the Gather protocol using current value $v_i^{r-1}$ 
            \State Store all accepted values in $V_i^{r-1} = \{v_j^{r-1}, j\in [n]\}$
            \State Let $k\coloneqq |V_i^{r-1}| -(n-t)$
            \State $v_i^r \leftarrow \ballMid(\safe_{n-2t+k}(V_i^{r-1}))$\label{line:asynch-ball}
        \EndFor
    \EndIndent
\end{algorithmic}
\end{algorithm}

\subsection{Correctness of Algorithm~\ref{alg:asynch-ballmidpoint}}

The following theorem states the correctness and runtime of Algorithm~\ref{alg:asynch-ballmidpoint}.

\begin{theorem}\label{thm:asynch-contraction}
    Algorithm~\ref{alg:asynch-ballmidpoint} solves MBAA according to Definition~\ref{def:MBAA} in the asynchronous communication model in $\log_2 \bigl(\frac{1}{\varepsilon}\cdot c\cdot\radius(\ballOf{C})\bigr)$ iterations when $n>\max(3,d+2)t$, where $C$ is the set of correct input vectors. The contraction rate in every iteration is upper bounded by $\frac{1}{\sqrt{2}}$.
\end{theorem}

We will prove Theorem~\ref{thm:asynch-contraction} by first focusing on the correctness. Note that the results from Lemma~\ref{lem:synch-safearea-nonempty} and Lemma~\ref{lem:synch-ballmidpoint-in-safearea} can be extended to the asynchronous setting: 

\begin{lemma}\label{lem:asynch-safearea-nonempty}
    The safe areas $\safe_{n-2t+k}(V_i^{r-1})$ computed by each node in Line~\ref{line:asynch-ball} of Algorithm~\ref{alg:asynch-ballmidpoint} are non-empty for $n>\max(3,d+2)t$. Additionally, $\safe_{n-2t+k}(V_i^{r-1})\subseteq\convexHull(C^{r-1})\ \forall i\in Corr$
\end{lemma}

The proof for this lemma can be derived analogously to the proof of Lemma~\ref{lem:synch-safearea-nonempty}. Note that the safe area computation used in Algorithm~\ref{alg:asynch-ballmidpoint} is analogous to the same step in~\cite{mendes2015multidimensional}. 

\begin{lemma}\label{lem:asynch-ballmidpoint-in-safearea}
    The ball midpoint $\ballMidOf{\safe_{n-2t+k}(V_i^{r-1})}$ computed by each node in Line~\ref{line:asynch-ball} of Algorithm~\ref{alg:asynch-ballmidpoint} is inside the safe area $\safe_{n-2t+k}(V_i^{r-1})$ of the respective node. 
\end{lemma}

This lemma is identical to Lemma~\ref{lem:synch-ballmidpoint-in-safearea} and thus the proof of Lemma~\ref{lem:synch-ballmidpoint-in-safearea} also applies here.
This concludes the convergence of Algorithm~\ref{alg:asynch-ballmidpoint}. In the following, we focus on $\varepsilon$-Agreement:

\begin{lemma}\label{lem:asynch-safe-in-trueball}
    All locally computed safe areas $\safe_{n-2t+k}(V_i^{r-1}), i\in Corr$ by the correct nodes in an iteration of Algorithm~\ref{alg:asynch-ballmidpoint} are inside the \seb of all correct vectors in that iteration. That is, $\safe_{n-2t+k}(V_i^{r-1})\subseteq\ballOf{C^{r-1}}.$
\end{lemma}

The proof of this lemma is analogous to the proof of Lemma~\ref{lem:safe-in-trueball}. Also the following lemma can be proven analogously to Lemma~\ref{lem:synch-diameter-included}:

\allowdisplaybreaks
\begin{lemma}\label{lem:asynch-diameter-included}
    Let $m$ be the midpoint of $\ballOf{C^{r-1}}$, and $m_i\coloneq \ballMidOf{\safe_{n-2t+k}(V_i^{r-1})}$ denote the midpoint of the \seb of the locally computed safe area of node $i$. Let $\overrightarrow{m,m_i}$ be the line going through the midpoint of the two balls. $\ballOf{C^{r-1}}$ includes the diameters of $\ballOf{\safe_{n-2t+k}(V_i^{r-1}))}$ that are orthogonal to $\overrightarrow{m,m_i}$.
\end{lemma}

In the following lemma, we show that all locally computed safe areas have a non-empty intersection also in the asynchronous case. Here, the proof differs from the synchronous setting, as the safe areas are computed in a slightly different way. 

\begin{lemma}\label{lem:asynch-all-safearea-intersect}
    Let $\safe_{n-2t+k_i}(V_i^{r-1}), i\in Corr$ be the locally computed safe areas by correct nodes in iteration $r-1$ of Algorithm~\ref{alg:asynch-ballmidpoint}, where $k_i\coloneqq |V_i^{r-1}| -(n-t)$. Then, 
    $$\bigcap_{i\in Corr}\safe_{n-2t+k_i}(V_i^{r-1})\neq \varnothing.$$
\end{lemma}
\begin{proof}
Consider iteration $r-1$ of the algorithm. By definition of Gather, there is a common core of at least $n-t$ vectors that is included in every local set $V_i^{r-1}$. Similar to the proof in the synchronous case, we define a global view $G^{r-1}$. In this case, we define the global view to  only contain $n-t$ vectors from the common core. If the common core has a larger number of vectors, we deterministically remove the vectors with the smallest IDs until $G^{r-1}$ only has $n-t$ vectors. Let $\safe_{n-2t}(G^{r-1})$ be the global safe area. Note that this set could potentially be the view of a correct node in the protocol, and thus $\safe_{n-2t}(G^{r-1})\neq\varnothing$. We will show that the global safe area is contained in every locally computed safe area.

Consider a correct node $i$ with the safe area $\safe_{n-2t+k_i}(V_i^{r-1})$. Let $U_\ell, \ell\in\left[\binom{n-t+k_i}{n-2t+k_i}\right]$ be subsets of $n-2t+k_i$ nodes used to compute the safe area of $V_i^{r-1}$, and $F_\lambda, \lambda\in\left[\binom{n-t}{n-2t}\right]$ be subsets used to compute the safe area of $G^{r-1}$. Let $D_i=V_i^{r-1}\setminus G^{r-1}$. Consider the sets $U_\ell\setminus D_i$. 
We have $|D_i|\le k_i$ as $V_i^{r-1}$ includes the common core and thus also $G^{k-1}$, and therefore $|U_\ell\setminus D_i|\ge n-2t$. 
Note that $\convexHull(U_\ell\setminus D_i)\subseteq \convexHull(U_\ell)\ \forall \ell$ as the convex hull is computed with potentially fewer vectors. On the other hand, for each set $U_\ell$, there exists a set $F_\lambda$, such that $F_\lambda\subseteq U_\ell\setminus D_i$. This is because $G^{r-1} = V_i^{r-1}\setminus D_i$ and $|V_i^{r-1}\setminus D_i|\ge n-2t$. Thus $\safe_{n-2t}(G^{r-1})=\bigcap_\lambda F_\lambda \subseteq \bigcap_\ell (U_\ell\setminus D_i) \subseteq \bigcap_\ell U_\ell=\safe_{n-2t+k_i}(V_i^{r-1})$. This concludes the proof.
\end{proof}

Given Lemma~\ref{lem:asynch-all-safearea-intersect}, the contraction rate can be proven:

\begin{lemma}\label{lem:asynch-contraction-rate}
    The contraction rate of Algorithm~\ref{alg:asynch-ballmidpoint} in every iteration is upper bounded by $\frac{1}{\sqrt{2}}$.
\end{lemma}

We omit the proof of this lemma here, as it is analogous to the proof of Lemma~\ref{lem:synch-contraction-rate}. The contraction lemma leads to the following corollary, which concludes the proof of Theorem~\ref{thm:asynch-contraction}.

\begin{corollary}
    Algorithm~\ref{alg:asynch-ballmidpoint} converges in $O\left(\log_2 \bigl(\frac{1}{\varepsilon}\cdot c\cdot\radius(\ballOf{C})\bigr)\right)$ such that the output vectors satisfy $\varepsilon$-Agreement.
\end{corollary}

\subsection{Tightness of the contraction bound for Algorithm~\ref{alg:asynch-ballmidpoint}}\label{sec:tightness-asynch}

In this section, we adapt the construction from Theorem~\ref{thm:lower-bound-synch} to show that the contraction rate is also asymptotically tight for the asynchronous algorithm. 

\begin{figure}[h]
    \centering
    \includegraphics[width=\textwidth]{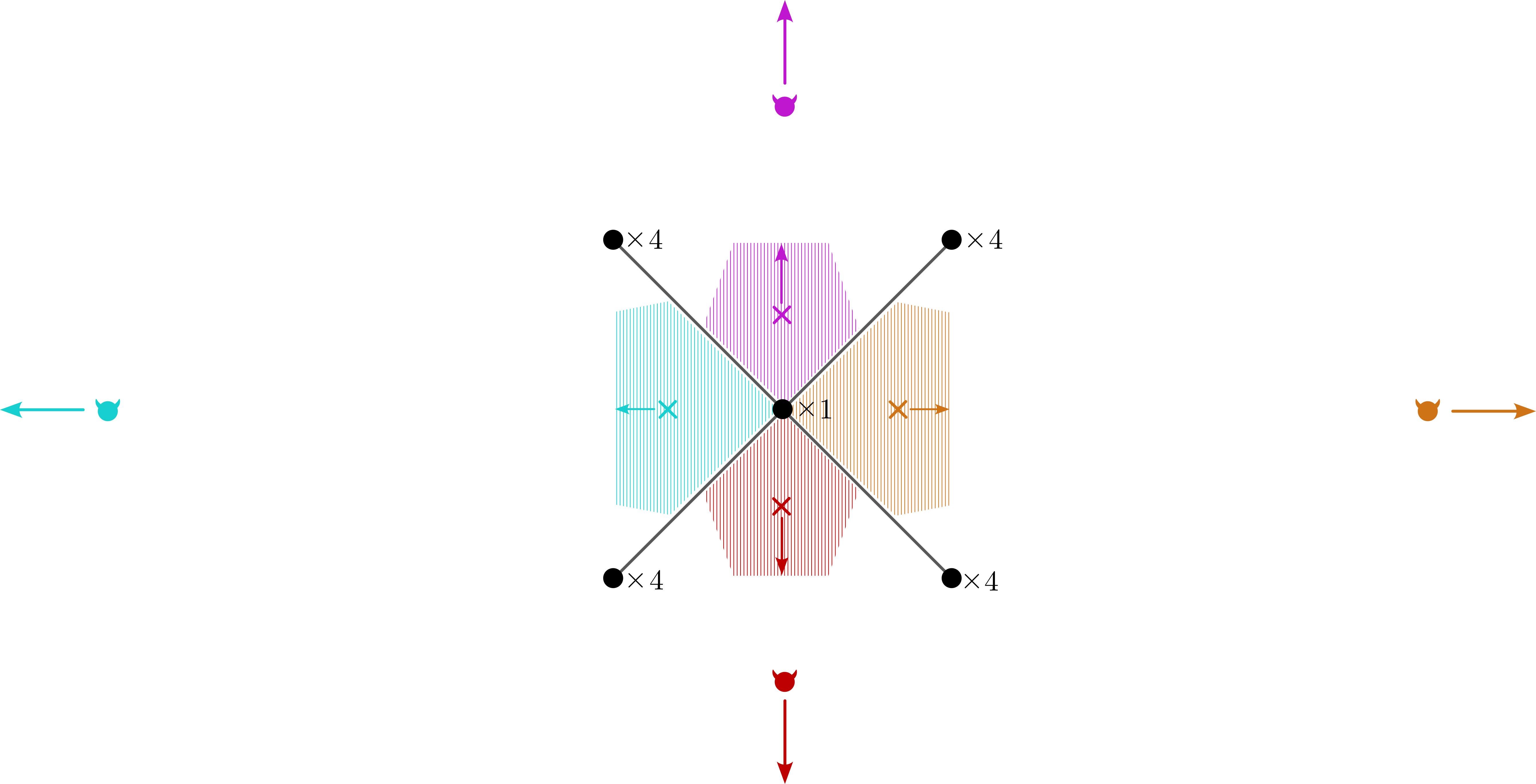}
    \caption{There are $17$ correct nodes and $4$ Byzantine nodes in this construction. $16$ correct vectors form a square, with one correct node in the middle, and all correct nodes will accept these vectors. Additionally, each group of $4$ correct nodes accepts one of the four Byzantine vectors. Observe that we could have omitted the correct vector in the middle of the square in this construction. The Byzantine vectors are assumed to be placed far away from the correct vectors in the marked direction. The further the placement of the Byzantine vectors, the closer the next input vector of the correct nodes to the midpoint of a side of the square. The safe areas corresponding to the four different views are marked in the figure. Note that new input vectors of the correct nodes for the next round are also forming a square from points that are close to the midpoints of the sides. Thus, Byzantine nodes can apply an analogous attack in the next iteration. }
    \label{fig:lower-bound-midpoint}
\end{figure}

\begin{theorem}\label{thm:lower-bound-asynch}
    The contraction rate of Algorithm~\ref{alg:asynch-ballmidpoint} is at least $\frac{1}{\sqrt{2}}-\delta$, where $\delta$ is an arbitrarily small constant.
\end{theorem}
\begin{proof}
    Consider a construction in $\mathbb{R}^2$ with $n-f\ge4f$ correct and $f$ Byzantine nodes, where $f\ge4$. Note that $n\ge 5f>3f$. For larger $d$, as in the synchronous case, we can consider vectors where all but two dimensions are $0$. We place input vectors of groups of $f$ correct nodes in the vertices of a square, with coordinates $\left(\frac{1}{\sqrt{2}},\frac{1}{\sqrt{2}}\right), \left(-\frac{1}{\sqrt{2}},\frac{1}{\sqrt{2}}\right), \left(\frac{1}{\sqrt{2}},-\frac{1}{\sqrt{2}}\right)$ and $\left(-\frac{1}{\sqrt{2}},-\frac{1}{\sqrt{2}}\right)$. Note that there are examples, where no correct nodes remain after this assignment. We therefore place the remaining $n-5f$ correct nodes at the origin, if applicable. Also in this construction, the \seb of the square is $1$. 
    
    For the adversarial strategy, as in the synchronous case, let four Byzantine nodes choose their inputs to be $(\alpha,0)$, $(-\alpha,0)$, $(0,\alpha)$ and $(0,-\alpha)$, where $\alpha>\max(\frac{1}{\delta},2)$. The remaining Byzantine nodes do not broadcast any input vectors, but may participate in the broadcast of other vectors. Further, let the adversarial scheduling divide the correct nodes into five views: four groups of $f$ correct nodes accept one Byzantine input vector, but not the three others; The remaining $n-5f$ correct nodes do not accept any Byzantine vectors. See Figure~\ref{fig:lower-bound-midpoint} for an example with $21$ nodes. 

    The above groups of $f$ nodes receive $n-f+1$ vectors and compute safe areas using subsets of $n-2f+1$ vectors. These safe areas are identical to the safe areas in the proof of Theorem~\ref{thm:lower-bound-synch}, and they all intersect in the origin. We refer to these safe areas as the main safe areas. The areas are bounded by the diagonals of the square, and, for $\alpha \rightarrow \infty$, they converge to triangles bounded by one side and the two diagonals of the square. The potentially remaining group of $n-5f$ correct nodes has a safe area that contains only the origin.

    The rest of the analysis, where the choice of $\alpha$ and the repeated construction in the following iteration is explained, is analogous to the proof of Theorem~\ref{thm:lower-bound-synch}.
\end{proof}

\section{Local round estimation}\label{sec:round-estimation}

Similar to~\cite{mendes2015multidimensional}, we apply the witness technique to estimate the number of rounds at the start of the algorithm. In the following, we differentiate between the synchronous and the asynchronous settings.

\paragraph*{Synchronous communication.}
The goal of the round estimation is to receive sets from other nodes that they use to compute the safe areas locally. By choosing the set whose \seb of the safe area is the largest, the nodes in the worst case over-estimate the radius of the \seb of the following iteration, but the value is still below the \seb of the initial inputs of the correct nodes. Algorithm~\ref{alg:synch-round-estimation} presents this idea as pseudocode.

\begin{algorithm}[tbh]
\caption{Preprocessing: Round estimation, synchronous}
\label{alg:synch-round-estimation}
\begin{algorithmic}[1]
    \State Each node $i \in [n]$ with input vector $v_i^0$ executes the following code:
    \Indent
        \State Broadcast $v_i^0$
        \State Receive $v_j^0$ from all nodes and store them in set $S_i$
        \State Broadcast $S_i$
        \State Accept a vector $v_j^0$, if it is in at least $n-t$ received sets $S_j$. \label{line:synch-round-accept} \Comment{all $n-f$ correct vectors will be accepted by $i$}
        \State Create a new set $T_i$ that contains all vectors that have been received in at least $n-2t$ sets $S_j$.
        \State Broadcast $T_i$
        \State Accept other sets $T_k$ if $T_k$ contains all accepted vectors from Line~\ref{line:synch-round-accept} \label{line:synch-round-accept-T}
        \State Set $W = \{\ballMid(\safe_{n-t}(T_k)): T_k\ \text{accepted}\}$
        \State Set $v_i^\prime =  \ballMid(\safe_{n-t}(W))$
        \State Set $rounds_i = \left\lceil \log_2(\frac{1}{\varepsilon}\cdot \radius(\ballOf{W})) \right\rceil$
        \State\Return $(v_i^\prime,rounds_i)$
    \EndIndent
\end{algorithmic}
\end{algorithm}

The idea of Algorithm~\ref{alg:synch-round-estimation} is to make sure that the sets $T_j$ that are used to estimate the radius are only accepted, if they contain all correct vectors.

\begin{lemma}\label{lem:correctness-synch-round-estimation}
    Algorithm~\ref{alg:synch-round-estimation} guarantees that every accepted set $T_j$ by a correct node contains all correct vectors. Additionally, every correct node will accept every set $T_j$ coming from a correct node. 
\end{lemma}
\begin{proof}
    Consider a correct node $i$. We start by proving that $T_k$ that $i$ accepts in Line~\ref{line:synch-round-accept-T} contains all correct vectors. Observe that every initially broadcast correct vector will be forwarded by every correct node, and since $n-f\ge n-t$, it also ends up in all correct sets $S_j$. Since every correct node receives every correct set $S_j$, every correct node will accept every correct vector. This ensures that, even if a set $T_b$ from a Byzantine adversary has been accepted, this set must have contained all $n-f$ correct vectors. 
    
    We now proceed with the proof that all correct sets $T_k$ will be accepted. Assume that some correct node $i$ accepted a vector $v_\ell^0$. Then, $v_\ell^0$ must have appeared in at least $n-t$ sets $S_m$ that $i$ received. At least $n-2t$ of these sets $S_m$ are from correct nodes and have been received by every correct node. Thus, all correct nodes must have added $v_\ell^0$ to their set $T_k$.
\end{proof}

In the following, we will use Lemma~\ref{lem:correctness-synch-round-estimation} to show that Algorithm~\ref{alg:synch-round-estimation} produces a suitable round estimation for every correct node.

\begin{theorem}
    Let $C^\prime$ be the set of output vectors $v_i^\prime$ computed by all correct nodes in Algorithm~\ref{alg:synch-round-estimation}. Then, no correct node underestimates the radius of $\ballOf{C^\prime}$, i.e., $\radius(\ballOf{W}) \ge \radius(\ballOf{C^\prime})$. Moreover, no node overestimates the radius of the ball of the correct input vectors $\ballOf{C}$, i.e., $\radius(\ballOf{C}) \ge \radius(\ballOf{W})$. Additionally $v_i^\prime$ are in the convex hull of all correct vectors: $v_i^\prime \in \convexHull(C)$.
\end{theorem}
\begin{proof}
    We first observe that $W$ always contains at least $n-f$ sets $T_k$ (Lemma~\ref{lem:correctness-synch-round-estimation}). Moreover, $\safe_{n-t}(W)\neq \varnothing$, since $n\ge|W|\ge n-f$ and by a Helly-type applied similarly to Section~\ref{sec:synch-MBAA}. Note that $v_i^\prime$ is inside $\convexHull(W)$. Therefore, $\convexHull(C^\prime)\subseteq\convexHull(W)$, from which $\radius(\ballOf{W}) \ge \radius(\ballOf{C^\prime})$ follows due to the minimality of $\ballOf{C^\prime}$. 

    For the second part, we show $\radius(\ballOf{C}) \ge \radius(\ballOf{W})$. The result from Lemma~\ref{lem:correctness-synch-round-estimation} says that all accepted sets $T_k$ by a correct node $i$ contain all correct input vectors. Thus, $\safe_{n-t}(T_k)$ is always inside $\convexHull(C)$, as one of the subsets of $n-t$ vectors used in the safe area computation consists only of correct vectors. Therefore, $\ballMid(\safe_{n-t}(T_k))\in\convexHull(C)$ Since this holds for every vector in $W$, $\convexHull(W) \subseteq \convexHull(C)$ and thus $\radius(\ballOf{C}) \ge \radius(\ballOf{W})$, and therefore the claim follows. 

    Finally, also $v_i^\prime \in \convexHull(C)$ follows, since all points in $W$ are also in $\convexHull(C)$ .

\end{proof}

Note that the contraction rate of the algorithm may be bigger than $\frac{1}{\sqrt{2}}$ in this preprocessing step. This step however only requires a small constant number of rounds which would not significantly influence the total convergence time of Algorithm~\ref{alg:synch-ballmidpoint}.

When this preprocessing step is used with Algorithm~\ref{alg:synch-ballmidpoint}, some nodes may terminate earlier than other. In this case, the nodes can broadcast their final vector and a halt message. This vector can be used as the input of the halted node for the following rounds, together with a set $S$, which coincides with the local set of the node.

\paragraph*{Asynchronous communication.}
In the asynchronous case, we need a stronger communication primitive to make sure that the correct nodes may rely on the sets of vectors that they have received from the Byzantine nodes. To this end, we will replace simple broadcast by reliable broadcast~\cite{BrachaRB}. Algorithm~\ref{alg:asynch-round-estimation} presents this strategy as a pseudocode.

\begin{algorithm}[tbh]
\caption{Preprocessing: Round estimation, asynchronous}
\label{alg:asynch-round-estimation}
\begin{algorithmic}[1]
    \State Each node $i \in [n]$ with input vector $v_i^0$ executes the following code:
    \Indent
        \State Reliably broadcast $v_i^0$
        \State Reliably receive $v_j^0$ from all nodes and store them in set $S_i$
        \State Reliably broadcast $S_i$
        \State Wait until $n-t$ sets $S_j$ with $|S_j|\ge n-t$ have been reliably received and accepted. A set is considered accepted if every vector in $S_j$ has been received reliably. 
        \State Set $W = \{\ballMid(\safe_{n-2t+k_j}(S_j)): S_j\ \text{accepted}\ k_j = |S_j|-(n-t)\}$
        \State Set $v_i^\prime =  \ballMid(\safe_{n-2t}(W))$
        \State Set $rounds_i = \left\lceil \log_2(\frac{1}{\varepsilon}\cdot \radius(\ballOf{W})) \right\rceil$
        \State\Return $(v_i^\prime,rounds_i)$
    \EndIndent
\end{algorithmic}
\end{algorithm}

Observe that Algorithm~\ref{alg:asynch-round-estimation} outputs a radius estimate and a new input vector that is used to execute Algorithm~\ref{alg:asynch-ballmidpoint}.

\begin{theorem}
    Let $C^\prime$ be the set of output vectors computed by all correct nodes in Algorithm~\ref{alg:asynch-round-estimation}. Then, no correct node underestimates the radius of $\ballOf{C^\prime}$, i.e., $\radius(\ballOf{W}) \ge \radius(\ballOf{C^\prime})$. Moreover, $\radius(\ballOf{C}) \ge \radius(\ballOf{W})$. Additionally, $v_i^\prime \in \convexHull(C)$ holds.
\end{theorem}
\begin{proof}
    We start by showing that no node underestimates the radius of $\ballOf{C^\prime}$, that is, $\radius(\ballOf{W}) \ge \radius(\ballOf{C^\prime})$ for the set $W$ of every correct node $i$. Note that $\safe_{n-2t}(W)\neq \varnothing$, because $n\ge|W|\ge n-t$ and by a Helly-type applied similarly to Section~\ref{sec:asynch-MBAA}. Observe next that $v_i^\prime$ is the midpoint of the ball of the safe area of $W$, meaning that $v_i^\prime\in \convexHull(W)$. By definition of $C^\prime$, $v_i^\prime \in C^\prime \subseteq \convexHull(C^\prime)$. Since  $v_i^\prime\in \convexHull(W)\ \forall i\in Corr$,  $\convexHull(C^\prime)\subseteq \convexHull(W)$.  Due to the minimality of $\ballOf{C^\prime}$ we also have $\radius(\ballOf{C^\prime}) \le \radius(\ballOf{W})$.
    
    Next, we show that the computed radius is upper bounded by $\radius(\ballOf{C})$.
    Observe that every $\ballMid(\safe_{n-2t+k_j}(S_j))$ is inside $\ballOf{C}$. This is because all the vectors inside the sets $S_j$ have been reliably received by node $i$, and because $|S_j|\ge n-t$, meaning that there is a subset of $n-2t+k_j$ correct vectors that have been used in the construction of the safe area. Since $\safe_{n-2t+k_j}(S_j)\subseteq \convexHull(C)$, the safe area is also inside $\ballOf{C}$, and with Lemma~\ref{lem:asynch-ballmidpoint-in-safearea} also its midpoint. This holds for all points in $W$. Therefore, $\radius(\ballOf{C}) \ge \radius(\ballOf{W})$.

    Finally, since all points in $W$ are also in $\convexHull(C)$, $\ballMid(\safe_{n-2t}(W)) = v_i^\prime$ is also inside $\convexHull(C)$.
\end{proof}

This theorem shows that no correct node will under-estimate the number of rounds, and thus the $\varepsilon$-Agreement property can be guaranteed. In the implementation, some correct nodes may run the algorithm longer than others. Therefore, once a node terminates, it needs to initiate or join a reliable broadcast with halt messages and only terminate once it reliably receives $n-t$ halt messages. During this time, a node should continue participating in the reliable communication routines with its final value. Since the safe areas make sure to stay inside the convex hull of all correct vectors of each preceding round, the nodes may not converge further, but they will terminate with $\varepsilon$-Agreement.

\section{Conclusion}

In this work, we proposed a novel algorithm for the MBAA problem in the synchronous and asynchronous communication models. We showed that a contraction rate of $\frac{1}{\sqrt{2}}$ can be achieved in both settings, and we presented a $2$-dimensional example, where the contraction bound is tight for the proposed algorithms. Finally, we showed how the nodes can locally estimate the number of rounds based on the given input distribution.

There are several interesting questions that remain for future work. The most natural question is whether the contraction rate can be improved even further. Note that in one dimension, the contraction rate is $1/2$ and it is tight. In two or more dimensions, there is still a gap between $1/2$ and $1/\sqrt{2}$ that should be investigated further. Another question is whether the algorithm can be extended to other communication models, such as the setting of dynamic networks, or the network-agnostic model. The main obstacle for both settings is to guarantee that all safe areas intersect under different network assumptions.

\section*{Acknowledgments}
This work was supported by the German Research Foundation (DFG), SPP 2378: ReNO-2 (grant 511099228), 2025-2029. The authors acknowledge that ChatGPT was used to generate counterexamples that ruled out several unsuccessful proof attempts for the contraction rate.

\bibliography{literature}

\end{document}